\documentclass[letterpaper,11pt]{article}
\usepackage{fullpage}
\usepackage{microtype}
\usepackage{times}
\usepackage{import}

\usepackage{amsmath,amsfonts,amssymb}
\usepackage[amsmath,amsthm,thmmarks,hyperref]{ntheorem}
\usepackage{mathtools}
\usepackage{hyperref}
\usepackage[capitalize,nameinlink,noabbrev]{cleveref} % must load after hyperref

\usepackage[inline]{enumitem}
\usepackage{xcolor}
\usepackage{breakcites}

\usepackage{algorithm, algorithmic}

\usepackage[multiuser,inline,nomargin,marginclue,draft]{fixme}

\usepackage{appendix}

\def\equationautorefname~#1\null{%
  Equation~(#1)\null
}

\usepackage{etoolbox}
\makeatletter
\patchcmd{\hyper@makecurrent}{%
    \ifx\Hy@param\Hy@chapterstring
        \let\Hy@param\Hy@chapapp
    \fi
}{%
    \iftoggle{inappendix}{%true-branch
        \@checkappendixparam{chapter}%
        \@checkappendixparam{section}%
        \@checkappendixparam{subsection}%
        \@checkappendixparam{subsubsection}%
        \@checkappendixparam{paragraph}%
        \@checkappendixparam{subparagraph}%
    }{}%
}{}{\errmessage{failed to patch}}

\newcommand*{\@checkappendixparam}[1]{%
    \def\@checkappendixparamtmp{#1}%
    \ifx\Hy@param\@checkappendixparamtmp
        \let\Hy@param\Hy@appendixstring
    \fi
}
\makeatletter

\newtoggle{inappendix}
\togglefalse{inappendix}

\apptocmd{\appendix}{\toggletrue{inappendix}}{}{\errmessage{failed to patch}}
\apptocmd{\subappendices}{\toggletrue{inappendix}}{}{\errmessage{failed to patch}}
\newcommand{\F}{\ensuremath{\mathbb{F}}}

\newcommand{\R}{\ensuremath{\mathbb{R}}}

\newcommand{\Z}{\ensuremath{\mathbb{Z}}}

\DeclarePairedDelimiter\inner{\langle}{\rangle}
\DeclarePairedDelimiter\abs{\lvert}{\rvert}
\DeclarePairedDelimiter\set{\{}{\}}
\DeclarePairedDelimiter\parens{(}{)}
\DeclarePairedDelimiter\round{\lfloor}{\rceil}
\DeclarePairedDelimiter\floor{\lfloor}{\rfloor}
\DeclarePairedDelimiter\ceil{\lceil}{\rceil}
\DeclarePairedDelimiter\length{\lVert}{\rVert}
\DeclarePairedDelimiterX\braket[2]{\langle}{\rangle}{#1 \delimsize\vert #2}
\newcommand{\matB}{\ensuremath{\mathbf{B}}}

\newcommand{\vecb}{\ensuremath{\mathbf{b}}}
\newcommand{\vecc}{\ensuremath{\mathbf{c}}}

\newcommand{\vect}{\ensuremath{\mathbf{t}}}

\newcommand{\vecv}{\ensuremath{\mathbf{v}}}
\newcommand{\vecw}{\ensuremath{\mathbf{w}}}

\newcommand{\vecy}{\ensuremath{\mathbf{y}}}
\newcommand{\vecz}{\ensuremath{\mathbf{z}}}
\newcommand{\veczero}{\ensuremath{\mathbf{0}}}

\usepackage{aliascnt}

\newtheorem{theorem}{Theorem}[section]

\def\NewTheorem#1#2{%
  \newaliascnt{#1}{theorem}
  \newtheorem{#1}[#1]{#2}
  \aliascntresetthe{#1}
  \expandafter\def\csname #1autorefname\endcsname{#2}
}

\theoremstyle{plain}            % following are theorem style

\NewTheorem{lemma}{Lemma}
\NewTheorem{corollary}{Corollary}
\NewTheorem{proposition}{Proposition}
\NewTheorem{claim}{Claim}
\NewTheorem{fact}{Fact}
\NewTheorem{maintheorem}{Main Theorem}

\theoremstyle{definition}       % following are def style

\NewTheorem{definition}{Definition}
\NewTheorem{conjecture}{Conjecture}
\NewTheorem{construction}{Construction}

\theoremstyle{remark}           % following are remark style

\NewTheorem{remark}{Remark}
\NewTheorem{example}{Example}
\NewTheorem{note}{Note}

\numberwithin{equation}{section}

\newcommand{\bit}{\ensuremath{\set{0,1}}}

\DeclareMathOperator{\poly}{poly}

\DeclareMathOperator*{\wt}{wt}

\DeclareMathOperator{\vol}{vol}
\renewcommand{\cal}[1]{\mathcal{#1}}

\makeatletter
\renewcommand{\pod}[1]{\mathchoice
  {\allowbreak \if@display \mkern 18mu\else \mkern 8mu\fi (#1)}
  {\allowbreak \if@display \mkern 18mu\else \mkern 8mu\fi (#1)}
  {\mkern4mu(#1)}
  {\mkern4mu(#1)}
}
\makeatletter
\let\@@pmod\pmod
\DeclareRobustCommand{\pmod}{\@ifstar\@pmods\@@pmod}
\def\@pmods#1{\mkern4mu({\operator@font mod}\mkern 6mu#1)}
\makeatother
\renewcommand{\epsilon}{\varepsilon}

\DeclareMathOperator{\KV}{KV} 

\FXRegisterAuthor{c}{ec}{\color{red}Chris}
\FXRegisterAuthor{e}{ee}{\color{blue}Ethan}
\fxusetheme{color}

\title{Lattice (List) Decoding Near Minkowski's Inequality}

\author{Ethan Mook\thanks{University of Michigan,
    \texttt{emook@umich.edu}.}
  \and Chris Peikert\footnote{Computer Science and Engineering,
    University of Michigan.  Email: \texttt{cpeikert@umich.edu}.  This
    material is based upon work supported by the National Science
    Foundation under Award CCF-2006857. The views expressed are those
    of the authors and do not necessarily reflect the official policy
    or position of the National Science Foundation. Part of this work
    was done while the author was visiting the Simons Institute Spring
    2020 program ``Lattices: Algorithms, Complexity, and
    Cryptography.''}}

\begin{document}
\maketitle

\begin{abstract}
  Minkowski proved that any $n$-dimensional lattice of unit determinant
has a nonzero vector of Euclidean norm at most~$\sqrt{n}$; in fact,
there are $2^{\Omega(n)}$ such lattice vectors. Lattices whose minimum
distances come close to Minkowski's bound provide excellent sphere
packings and error-correcting codes in~$\R^{n}$.

The focus of this work is a certain family of efficiently
constructible $n$-dimensional lattices due to Barnes and Sloane, whose
minimum distances are within an $O(\sqrt{\log n})$ factor of
Minkowski's bound. Our primary contribution is a polynomial-time
algorithm that \emph{list decodes} this family to distances
approaching $1/\sqrt{2}$ of the minimum distance. The main
technique is to decode Reed-Solomon codes under error measured in the
Euclidean norm, using the Koetter-Vardy ``soft decision'' variant of
the Guruswami-Sudan list-decoding algorithm.

%%% Local Variables:
%%% mode: latex
%%% TeX-master: "bch-lattice-decoding"
%%% End:

\end{abstract}

% list notes to address; appears only in draft mode
\listoffixmes

\section{Introduction}
\label{sec:introduction}

A linear (error-correcting) \emph{code}~$\cal{C}$ is a linear subspace
of~$\F_{q}^{n}$, and its \emph{minimum distance}~$d = d(\cal{C})$ is
the minimum Hamming weight of its nonzero code words. Such a code is
uniquely decodable under error having weight up to $\frac{d-1}{2}$,
and for many codes it is known how to perform such decoding
efficiently. See, e.g.,~\cite{GRS:_essential_coding_theory} for
extensive background on the combinatorial and algorithmic aspects of
codes.

Elias~\cite{elias-zero-error} and Wozencraft~\cite{wozencraft} put
forward the idea of decoding under error having weight~$d/2$ or more,
which can potentially lead to ambiguity. The goal of \emph{list
  decoding} is to find \emph{all} codewords within a certain distance
of a received word, and it is often possible to guarantee that there
are not too many. Breakthrough list-decoding algorithms were obtained
by Goldreich and Levin~\cite{DBLP:conf/stoc/GoldreichL89} for the
Hadamard code, and by Sudan~\cite{DBLP:journals/jc/Sudan97} and
Guruswami and Sudan~\cite{DBLP:journals/tit/GuruswamiS99} for
Reed-Solomon codes. These results and others have had countless
applications across computer science.

A (full-rank) \emph{lattice}~$\Lambda \subset \R^{n}$ is a discrete
additive subgroup whose linear span is~$\R^{n}$, and its minimum
distance $\lambda_{1} = \lambda_{1}(\Lambda)$ is the minimum Euclidean
norm of its nonzero lattice points. Naturally, a lattice is uniquely
decodable under error having norm less than $\lambda_{1}/2$, and
efficient algorithms are known for decoding some lattices up to that
distance. For example, for the integer lattice~$\Z^{n}$, which has
unit minimum distance, we can simply round each coordinate to the
nearest integer.

A main measure of a lattice's ``quality,'' e.g., as a sphere packing
or error-correcting code, is its normalized minimum distance
$\sqrt{\gamma(\Lambda)} =
\lambda_{1}(\Lambda)/\det\parens{\Lambda}^{1/n}$ (the square root is
present for historical reasons), where the determinant
$\det(\Lambda) = \vol(\R^{n}/\Lambda)$ is the covolume of the lattice,
i.e., the absolute value of the determinant of any $\Z$-basis
of~$\Lambda$.  A theorem of Minkowski bounds this quality by
$\sqrt{\gamma(\Lambda)} \leq \sqrt{n}$ for any lattice
$\Lambda \subseteq \R^{n}$; in fact, a refined version says that there
are an exponential $2^{\Omega(n)}$ number of lattice vectors of norm
at most $\sqrt{n} \cdot \det\parens{\Lambda}^{1/n}$. There exist
infinite families of lattices whose minimum distances are
asymptotically tight with Minkowski's bound (up to a small constant
factor), and there are efficiently constructible families that are
nearly tight with it (see~\cite{conwaysloane99:_splag} for extensive
background, and below for further details).

\paragraph{Lattice (list) decoding.}

Motivated by the many similarities between codes and lattices,
Grigorescu and Peikert~\cite{DBLP:conf/coco/GrigorescuP12} initiated
the study of (efficient) list decoding for lattices, and, building on
the unique-decoding algorithm of Micciancio and
Nicolosi~\cite{DBLP:conf/isit/MicciancioN08}, gave an algorithm for
the well known Barnes-Wall family of
lattices~$\text{BW}_{k} \subset \R^{n}$, where $n=2^{k}$. Their
algorithm's running time, and hence list size, is polynomial in~$n$
for decoding distances approaching the minimum
distance~$\lambda_{1}(\text{BW}_{k}) = \Theta(\sqrt{n})$ (but at this
threshold and beyond, the list size can be super-polynomial). However,
Barnes-Wall lattices have only moderately good quality: compared to
Minkowski's inequality, they satisfy the much tighter bound
$\sqrt{\gamma(\text{BW}_{k})} = O(n^{1/4})$. So, the results
of~\cite{DBLP:conf/coco/GrigorescuP12} are quite far from optimal in
terms of (determinant-normalized) decoding distance. By contrast, a
recent work of Ducas and Pierrot~\cite{DBLP:journals/dcc/DucasP19}
gave, for a certain family of lattices, a simple and efficient
decoding algorithm for normalized distance $\Theta(\sqrt{n}/\log n)$,
which is tight with Minkowski's bound up to an $O(\log n)$
factor. However, their algorithm only performs \emph{unique} (not
list) decoding, below half the minimum distance.

In this work we focus on an infinite family of lattices
$\Lambda_{n} \subset \R^{n}$, originally constructed by Barnes and
Sloane~\cite{barnes83:_lattice_packings}, having normalized minimum
distance $\sqrt{\gamma(\Lambda_{n})} = \Omega(\sqrt{n/\log n})$. Our
main contribution is a polynomial-time list-decoding algorithm for
this family, which decodes to distance
nearly~$\lambda_{1}(\Lambda_{n})/\sqrt{2}$. An immediate corollary is
the first (to our knowledge) polynomial-time \emph{unique} decoder for
this family, for up to half the minimum distance. In summary, we
obtain efficient (list-)decoding algorithms for distances within an
$O(\sqrt{\log n})$ factor of the universal barrier implied by
Minkowski's theorem (since for any lattice there can be exponentially
many lattice points within normalized distance~$\sqrt{n}$ of a target
point).

\begin{theorem}
  \label{thm:main}
  There is an efficiently constructible family of
  lattices~$\Lambda = \Lambda_{n} \subset \R^{n}$ having normalized
  minimum distance
  $\lambda_{1}(\Lambda)/\det\parens{\Lambda}^{1/n} =
  \Omega(\sqrt{n/\log n})$, which for any constant $\varepsilon > 0$
  are list decodable to within distance
  $(1/\sqrt{2}-\varepsilon) \cdot \lambda_{1}(\Lambda)$ in some
  $\poly(n)$ time.
\end{theorem}

We remark that $\lambda_{1}(\Lambda)/\sqrt{2}$ is a natural generic
barrier for (combinatorially) efficient list decoding on lattices. On
the one hand, Rankin's bound implies that for any lattice
$\Lambda \subset \R^{n}$ and target point, there are at most~$2n$
lattice points within distance $\lambda_{1}(\Lambda)/\sqrt{2}$ of the
target (this is analogous to Johnson's bound for codes). On the other
hand, for any constant $\epsilon > 0$ there exists a family of
``locally dense'' lattices $\Lambda_{n} \subset \R^{n}$ and target
points $\vect_{n} \in \R^{n}$ for which there are
$\exp(n^{\Omega(1)})$ points in~$\Lambda_{n}$ within distance
$(1/\sqrt{2}+\epsilon) \cdot \lambda_{1}(\Lambda_{n})$
of~$\vect_{n}$~\cite{DBLP:journals/siamcomp/Micciancio00}. In fact, as
shown in~\cite{DBLP:journals/toc/Micciancio12}, the family of lattices
from \cref{thm:main} has this property, but with~$1/\sqrt{2}$ replaced
by~$\sqrt{2/3}$. It is an interesting question whether constant
factors smaller than~$\sqrt{2/3}$ for local density, or larger than
$1/\sqrt{2}$ for efficient list decoding, can be obtained for this
family.

Finally, we point out that the lattices from \cref{thm:main} also have
efficiently constructible \emph{short} bases, consisting of vectors
whose norms are within a~$\sqrt{n}$ factor of the minimum distance
(see \cref{rem:construction-d-short-basis}). This can be useful for
\emph{encoding}, where one typically wants to map a message to a
relatively short lattice point: its norm corresponds to the power
required to send it, and one wants to minimize the power-to-noise
ratio. By contrast, for the uniquely decodable lattices studied
in~\cite{DBLP:journals/dcc/DucasP19}, relatively short nonzero lattice
vectors are not known, and may even be infeasible to compute, so it is
unclear whether power-efficient encoding is possible.

\paragraph{Techniques and organization.}

The family of lattices from \cref{thm:main} is obtained by applying
``Construction~D''~\cite{conwaysloane99:_splag} to a certain tower of
BCH codes, which are subfield subcodes of certain Reed-Solomon
codes. It was shown
in~\cite{barnes83:_lattice_packings,DBLP:journals/toc/Micciancio12}
that this family of lattices yields an excellent packing, with
normalized minimum distance $\Omega(\sqrt{n/\log n})$. We obtain an
efficient list-decoding algorithm for this family as follows.

First, in \cref{sec:construction-a} we give an efficient list-decoding
algorithm for any prime-subfield subcode
$\cal{C} = \cal{R} \cap \F_{p}^{n}$ of a Reed-Solomon
code~$\cal{R} \subseteq \F_{p^{r}}^{n}$ (which in particular includes
the BCH codes mentioned above), where error is measured in the
\emph{Euclidean} (rather than Hamming) norm. More specifically,
embedding the code~$\cal{C} \subseteq \F_{p}^{n}$ into $(\R/p\Z)^{n}$
in the natural way, the algorithm recovers all codewords of~$\cal{C}$
that are within squared Euclidean norm nearly $d/2$ of an arbitrary
received word in $(\R/p\Z)^{n}$, where~$d$ is the minimum
\emph{Hamming} distance of~$\cal{R}$ (and hence is a lower bound on
the minimum distance of~$\cal{C}$). The core of our algorithm is the
Koetter-Vardy~\cite{DBLP:journals/tit/KoetterV03a} ``soft-decision''
list decoder for Reed-Solomon codes, which takes as one of its inputs
a \emph{reliability vector} representing the probability of each field
element at each position of the received word. However, in our setting
there are no probabilities, just a fixed received word in
$(\R/p\Z)^{n}$. We use it to define an \emph{implicit} reliability
vector, and show that with this vector the Koetter-Vardy algorithm
recovers all codewords within the desired Euclidean norm. In fact,
using the framework of~\cite{KoetterV01:_optimal_weight} we show that
our choice of reliability vector yields an optimal tradeoff between
code dimension and Euclidean decoding distance for the analysis of the
Koetter-Vardy soft-decision decoder, for squared distances at most
$n/4$.

Next, in \cref{sec:construction-d} we give a list-decoding algorithm
for any Construction D lattice defined by a tower of codes over a
prime field~$\F_{p}$, using subroutines that list decode the component
codes to appropriate distances (in any~$\ell_{q}$ norm, including the
Euclidean norm). The algorithm naturally arises from the iterative
definition of Construction D, and works by recursively recovering each
nearby lattice vector from its least- to most-significant digit in
base~$p$. (A similar algorithm for \emph{unique} decoding of
Construction D lattices was studied
in~\cite{DBLP:conf/globecom/MatsumineKO18}, with a focus on simulation
in moderate dimensions for BCH codes, but no theorems about its
behavior were given.) The depth of the recursion is the number of
codes in the tower, and the branching factor at each level of the
recursion is the size of the list output by the decoding subroutine at
that level.

Finally, in \cref{sec:near-minkowski} we recall the lattice family
obtained by instantiating Construction D with a certain tower of BCH
codes, and instantiate the subroutines in the list-decoding algorithm
from \cref{sec:construction-d} with the algorithm from
\cref{sec:construction-a}. Because a tower of $\Theta(\log n)$ codes
is needed to obtain the desired density, to obtain a polynomial-time
algorithm we need to ensure constant $O(1)$ list sizes for each code
in the tower. This allows us list decode the lattice to within a
$1/\sqrt{2}-\varepsilon$ factor of its minimum distance in the
Euclidean norm, for any constant $\varepsilon > 0$.

\paragraph{Acknowledgments.}

We thank Daniele Micciancio and the anonymous reviewers for helpful
comments on the presentation.

%%% Local Variables:
%%% mode: latex
%%% TeX-master: "bch-lattice-decoding"
%%% End:

\section{Preliminaries}
\label{sec:preliminaries}

For real vectors
$\vecv = (v_{1}, \ldots, v_{n}), \vecw = (w_{1}, \ldots, w_{n}) \in
\R^{n}$, define their inner product as
$\inner{\vecv, \vecw} := \sum_{i=1}^{n} v_{i}w_{i}$. For any positive
integer~$p$, let $\Z_{p} := \Z/p\Z$ denote the quotient group of the
integers modulo~$p$. When~$p$ is prime, we identify~$\Z_{p}$ with the
finite field~$\F_{p}$ in the obvious way. Fixing some arbitrary set of
representatives of~$\Z_{p}$ (e.g., $\set{0,1,\ldots,p-1}$), for
$v \in \Z_{p}$ let $\overline{v} \in \Z$ denote its representative,
and extend this notation coordinate-wise to vectors over~$\Z_{p}$.
Similarly, let $\R_{p} := \R/p\Z$ be the quotient group of the real
numbers modulo the integer multiples of~$p$; then $\Z_{p}$ is a
subgroup of~$\R_{p}$.

For $y \in \R_{p}$, we define
$\abs{y} := \min \set{ \abs{z} : z \in \R \cap (y + p\Z)}$ to be the
minimal absolute value over all real numbers congruent to~$y$
(modulo~$p\Z$). Equivalently, it is the absolute value of the single
element in $[-p/2,p/2) \cap (y + p\Z)$. For
$\vecy = (y_{1}, \ldots, y_{n}) \in \R_{p}^{n}$, we define the
Euclidean norm
\[ \length{\vecy} := \parens[\Big]{\sum_{i=1}^{n}
    \abs{y_{i}}^{2}}^{1/2} = \min \set {\length{\vecz} : \vecz \in
    \R^{n} \cap (\vecy + p\Z^{n})} . \]

\paragraph{Error-correcting codes.}

For a prime power~$q$ and nonnegative integers $k \leq n \leq q$, a
\emph{Reed-Solomon code}~\cite{ReedSolomon} of length~$n$ and
dimension~$k$ over~$\F_{q}$ is the set
\[ \cal{R}_{\F_{q}}[n,k] := \set{(p(\alpha_{1}), p(\alpha_{2}),
    \ldots, p(\alpha_{n})) \in \F_{q}^{n} : p(X) \in \F_{q}[X],
    \deg(p) < k} \] for some fixed distinct evaluation points
$\alpha_{1}, \ldots, \alpha_{n} \in \F_{q}$.  It is easy to see
that~$\cal{R}$ is a linear code (i.e., a linear subspace), and that
its minimum Hamming distance
$d := \min_{\vecc \in \cal{R}} \wt(\vecc) = n-k+1$, because a
polynomial~$p(X) \in \F_{q}[X]$ can have at most $\deg(p)$ zeros (and
this is attainable when $\deg(p) \leq q$). It is also clear that for
any $k \leq k'$ we have $\cal{R}[n, k] \subseteq \cal{R}[n, k']$. For
our purposes it is convenient to define the adjusted rate
$R^{*} = (k-1)/n = 1-d/n$.

If~$q$ is a power of a prime~$p$, making~$\F_{p}$ a subfield
of~$\F_{q}$, \emph{BCH
  codes}~\cite{BoseRayChaudhuri,hocquenghem59:_codes} can be obtained
as $\F_{p}$-subfield subcodes of certain Reed-Solomon codes. More
specifically, letting $\F_{q}^{*} = \F_{q} \setminus \set{0}$ be the
set of evaluation points, the (primitive, narrow-sense) BCH code of
designed distance~$1 \leq d \leq n$ is defined as
\[ \cal{C}_{\F_{q}}[n=q-1,d] := \cal{R}_{\F_{q}}[n,k=n-d+1] \cap
  \F_{p}^{n}. \] Clearly, $\cal{C}_{\F_{q}}[n,d]$ has minimum distance
at least~$d$ (because it is a subset of a distance-$d$ code). It
follows from the nesting property of Reed-Solomon codes that for any
$d \leq d'$ we have
$\cal{C}_{\F_{q}}[n, d] \supseteq \cal{C}_{\F_{q}}[n, d']$.  It is
also known that $\cal{C}_{\F_{q}}[n,d]$ is efficiently constructible,
in the sense that an $\F_{p}$-basis for it can be produced in time
$\poly(n)$. Finally, its dimension~$k$ satisfies the following well
known bound (see, e.g.,
\cite[Exercise~5.10]{GRS:_essential_coding_theory}).

\begin{lemma}
  \label{lem:bch-dimension}
  For $1 \leq d \leq n = q-1$, the BCH code $\cal{C}_{\F_{q}}[n,d]$
  has codimension $n-k \leq \ceil{\frac{p-1}{p} (d-1)} \log_{p} q$.
\end{lemma}

\paragraph{Lattices.}

A \emph{lattice}~$\Lambda$ is a discrete additive subgroup
of~$\R^{n}$. If its linear span is~$\R^{n}$, the lattice is said to be
\emph{full rank}; from now on, we limit our attention to such
lattices. Any lattice is generated (non-uniquely) as the integer
linear combinations of the vectors in a \emph{basis}
$\matB = \set{\vecb_{1}, \ldots, \vecb_{n}}$, as
$\Lambda = \set{\sum_{i=1}^{n} z_{i} \vecb_{i} : z_{i} \in \Z}$. The
\emph{minimum distance}
$\lambda_{1}(\Lambda) := \min_{\vecv \in \Lambda \setminus
  \set{\veczero}} \length{\vecv}$ is the length of any shortest
nonzero lattice vector. The \emph{determinant}
$\det(\Lambda) := \vol(\R^{n}/\Lambda) = \abs{\det(\matB)}$ is the
absolute value of the determinant of any basis of the lattice.

A standard way of measuring the ``density'' of a lattice is to
normalize its minimum distance by its (dimension-adjusted)
determinant. More specifically, for a (full-rank) lattice
$\Lambda \subset \R^{n}$, its normalized minimum distance is
$\lambda_{1}(\Lambda)/\det(\Lambda)^{1/n}$. A simple application of a
theorem of Minkowski yields the inequality
$\lambda_{1}(\Lambda) \leq \sqrt{n} \cdot \det(\Lambda)^{1/n}$, which
is often called Minkowski's bound. This bound is known to be tight up
to a constant factor, i.e., there exists an infinite family of
$n$-dimensional lattices $\Lambda_{n}$ of unit determinant for which
$\lambda_{1}(\Lambda) = \Omega(\sqrt{n})$.

% \begin{definition}[Residue decoder]
%   Let $\Lambda' \subseteq \Lambda \subset \R^n$ be lattices. A residue
%   decoder for radius~$d$ on $\Lambda/\Lambda'$ is an algorithm that
%   takes as input a coset $\vecy + \Lambda \in \R^n/\Lambda$ and
%   outputs a list of all cosets $\vecy' + \Lambda'$ such that
%   $\vecy' = \vecy \pmod*{\Lambda}$ and
%   $\lambda_{1}(\vecy' + \Lambda') \leq d$.
% \end{definition}

% \cnote{Explanation of trivial equivalence between (residue) coset
%   decoder and ``near vector'' decoder.}

% The decoder for~$\Lambda$ works as follows. Given a received word
% $\vecy \in \R^{n}$, first $\cal{D}(\vecy+\Lambda)$ outputs all cosets
% $\vecy' + \Lambda'$ such that $\vecy' = \vecy \pmod*{\Lambda}$ and
% $\lambda_{1}(\vecy' + \Lambda') \leq d$. Then
% $\cal{D}'(\vecy' + \Lambda')$ outputs all $\vece \in \R^{n}$ such that
% $\vece = \vecy' \pmod*{\Lambda'}$ and $\length{\vece} \leq
% d$. Conclusion: $\vecv = \vecy - \vece \in \Lambda$ and
% $\dist(\vecy,\vecv) = \length{\vece} \leq d$.

% \begin{theorem}
%   \label{thm:residue-to-normal}
%   Let $\Lambda' \subseteq \Lambda \subseteq \R^n$ be lattices,
%   let~$\cal{D}$ be a residue decoder for radius~$d$ on
%   $\Lambda/\Lambda'$, and let~$\cal{D}'$ be a decoder for radius~$d$
%   on~$\Lambda'$.
% \end{theorem}

%%% Local Variables:
%%% mode: latex
%%% TeX-master: "bch-lattice-decoding"
%%% End:

\section{Decoding Reed-Solomon Subfield Subcodes in the Euclidean Norm}
\label{sec:construction-a}

In this section we give a list decoder, for error measured in the
\emph{Euclidean norm}, for the $\F_{p}$-subfield
subcode~$\cal{C} = \cal{R} \cap \F_{p}^{n}$ (for a prime~$p$) of any
Reed-Solomon code~$\cal{R} = \cal{R}_{\F_{q}}[n,k > 1]$, where
$q=p^{r}$ for some $r \geq 1$. (In particular, this includes BCH
codes.) More specifically, given a received word in $\R_{p}^{n}$, the
decoder outputs all codewords in~$\cal{C}$ that are within Euclidean
norm nearly $\sqrt{d/2}$ of the received word, where~$d = n-k+1$ is
the minimum Hamming distance of the Reed-Solomon code (and hence a
lower bound on the minimum Hamming distance of~$\cal{C}$).

The heart of our algorithm is the ``soft-decision'' list-decoding
algorithm of Koetter and Vardy~\cite{DBLP:journals/tit/KoetterV03a}
for Reed-Solomon codes, which uses data about the \emph{likelihood} of
each alphabet symbol in each position. More precisely, to decode a
length-$n$ Reed-Solomon code over~$\F_{q}$, their algorithm takes as
one of its inputs a so-called \emph{reliability vector}
$\Pi \in [0,1]^{qn}$. Such a vector consists of~$n$ length-$q$ blocks,
where the $j$th entry of the $i$th block represents the probability
that the transmitted codeword had the~$j$th element of~$\F_{q}$ in
its~$i$th coordinate. That is, each of the length-$q$ blocks in $\Pi$
is a probability mass function, and in particular has unit~$\ell_{1}$
norm.

In our setting, we have no explicit probabilities of transmitted
symbols, only a received word $\vecy \in \R_{p}^{n}$. In
\cref{sec:reliability-vectors} we define a mapping that converts the
received word to an \emph{implicit} reliability vector, which we
provide to the Koetter-Vardy soft decoder. As we show in
\cref{sec:decoding-algorithm}, this choice of reliability vector
allows the decoder to find all codewords within a desired Euclidean
distance of the received vector.\footnote{It seems likely that this
  approach generalizes somewhat to Reed-Solomon codes themselves (over
  prime-power fields). We restrict our attention to prime-subfield
  subcodes because they admit a natural Euclidean norm, and they are
  required for defining Construction D lattices (see
  \cref{sec:construction-d}).} Moreover, in
\cref{sec:optimal-reliability} we show that our choice of reliability
vector is essentially optimal for the range of decoding distances that
are relevant to this work (and even somewhat beyond).

% \cnote{For future investigation: it seems that we can generalize the
%   reliability vector to map $\F_{p^{r}} \cong \F_{p}^{r}$ to indicator
%   vectors in $\R^{p^{r}}$, and extend to a map from $\R_{p}^{r}$ to
%   $\R^{p^{r}}$ by interpolating among the $2^{r}$ ``nearby'' points in
%   $\F_{p}^{r}$, with weights given by the product of 1-distance in
%   each coordinate (so that they sum to 1, as needed). Then it seems
%   that \cref{lem:embedding} may hold (still with coefficient 2), but
%   \cref{thm:inner-product-condition} changes to have an extra
%   $\sqrt{2^{r}} \leq \sqrt{p^{r}}$ factor in~$S$ (roughly). This may
%   not be a problem for some applications, e.g., if $r$ is constant.}

\subsection{Reliability Vectors}
\label{sec:reliability-vectors}

We now define the mapping from received words to reliability
vectors. For $c \in \F_{q} \setminus \F_{p}$ define
$[c] = \veczero \in [0,1]^{p}$, and for $c \in \F_p$ define
$[c] \in \bit^p$ to be the indicator vector of~$c$, i.e., the entry
indexed by~$c$ is~1 and all other entries are~0. We extend this
definition to~$\R_p$ by mapping each interval
$(c, c+1) \subseteq \R_{p}$ to the open line segment
$\set{(1 - t)[c] + t[c + 1] : t \in (0, 1)} \subseteq [0,1]^{p}$ in
the natural way. In other words, for $y \in [c, c+1]$ we define
\begin{equation}
  \label{eq:bracket-y}
  [y] := [c] + ([c + 1] - [c]) \cdot \abs{y - c}.
\end{equation}
We extend the notation $[\cdot]$ to vectors by applying it entry-wise,
mapping $n$-dimensional vectors to~$[0,1]^{pn}$.\footnote{This can
  equivalently be viewed as outputting a matrix in
  $[0,1]^{p \times n}$, which is the perspective used more
  in~\cite{DBLP:journals/tit/KoetterV03a}, but the vector view will be
  more natural for us.}  Following the terminology
of~\cite{DBLP:journals/tit/KoetterV03a}, we call
$[\vecy] \in [0,1]^{pn}$ the \emph{reliability vector} of a received
word~$\vecy \in \R_{p}^{n}$.

\begin{lemma}
  \label{lem:embedding}
  For any $\vecy \in \R_p^{n}$ and $\vecc \in \Z_p^{n}$, we have
  \begin{equation}
    \length{[\vecy] - [\vecc]}^2 \leq 2 \length{\vecy - \vecc}^2 .
  \end{equation}
\end{lemma}

\begin{proof}
  It suffices to prove the lemma for $n = 1$ because both sides split
  as sums over their components. Let $y \in \R_p$ and $c \in \Z_p$. If
  $\abs{y - c} \geq 1$, then the nonzero entries of
  $[y],[c] \in [0,1]^{p}$ are in distinct positions, hence
  $\length{[y] - [c]}^2 = \length{[y]}^{2} + \length{[c]}^{2} \leq
  \length{[y]}^{2} + 1 \leq 2$. Otherwise $\abs{y - c} < 1$; assume
  that $y \in [c, c + 1]$. Then by \cref{eq:bracket-y},
  \begin{equation}
    \length{[y] - [c]}^2 =
    \length{[c + 1] - [c]}^2 \cdot \abs{y - c}^2 = 2\abs{y - c}^2.
  \end{equation}
  The case $y \in [c-1,c]$ proceeds symmetrically.
\end{proof}

\subsection{Decoding Algorithm}
\label{sec:decoding-algorithm}

Here we define and analyze our decoding algorithm, presented in
\cref{alg:euclidean-bch-decoder}. There and in what follows, we let
$\KV(\Pi, S)$ denote the Koetter-Vardy soft-decision decoder for
$\cal{R}$, run on input reliability vector $\Pi \in [0,1]^{pn}$ with
output list size limited to~$S$.\footnote{To be completely accurate,
  the KV algorithm as defined in~\cite{DBLP:journals/tit/KoetterV03a}
  takes a reliability vector in $[0,1]^{qn}$; we can either
  appropriately pad the $[0,1]^{pn}$-vector with zeros, or, to be more
  efficient, modify the algorithm to work directly with
  $[0,1]^{pn}$-vectors in the obvious way.}

Although we use the Koetter-Vardy algorithm essentially as a `black
box,' we include the following high-level description of its operation
for completeness. First, it defines a `multiplicity vector'
$M = \floor{\lambda \Pi} \in \Z^{pn}$ for some suitably large scaling
factor~$\lambda \in \R^{+}$, which is determined based on the desired
list size bound~$S$. It then proceeds by a generalization of the
list-decoding algorithm of Guruswami and
Sudan~\cite{DBLP:journals/tit/GuruswamiS99}. More specifically, it
uses~$M$ to set up a system of linear equations, which it solves to
compute the minimal bivariate polynomial $\mathcal{Q}_{M}(X,Y)$ having
zeroes with multiplicities given by~$M$ at specified points. Finally,
it (partially) factors~$\mathcal{Q}_{M}$ to identify all the factors
of the form $Y-f(X)$, which directly correspond to the output list of
codewords. The running time of the algorithm is primarily determined
by the number of equations in the linear system, which is
asymptotically the sum of the squares of the entries of~$M$.

\begin{algorithm}
  \caption{List-decoding algorithm for code $\cal{C}$, for the
    Euclidean norm}
  \label{alg:euclidean-bch-decoder}
  \begin{algorithmic}
    \REQUIRE Received word $\vecy \in \R_{p}^n$ and $\varepsilon > 0$.

    \ENSURE A list of the codewords $\vecc \in \cal{C}$ for which
    $\length{\vecy - \vecc}^{2} \leq (1 - \varepsilon)d/2$.

    \begin{enumerate}
    \item Let $L = \KV([\vecy], S) \subseteq \cal{R}$ be the output
      list of the soft-decision decoder of~$\cal{R}$ on reliability
      vector~$[\vecy]$, with list size limited to
      \begin{equation}
        \label{eq:list-size-bound}
        S := \frac{\frac{1}{R^*} + \frac{1}{\sqrt{2R^*}}}
        {1 - \sqrt{\frac{R^*}{\varepsilon + (1-\varepsilon)R^{*}}}}.
      \end{equation}
    \item Output $\set{\vecc \in L \cap \F_{p}^{n} :
        \length{\vecy-\vecc}^{2} \leq (1-\varepsilon)d/2}$.
    \end{enumerate}
  \end{algorithmic}
\end{algorithm}

\begin{theorem}[{{Adapted from \cite[Theorem~17]{DBLP:journals/tit/KoetterV03a}.}}]
  \label{thm:inner-product-condition}
  Let $\vecy \in \R_{p}^{n}$ and $S > 0$. The soft-decision decoder
  $\KV([\vecy], S)$ for~$\cal{R}$ outputs a list of all the at
  most~$S$ codewords $\vecc \in \cal{R}$ such that
  \begin{equation}
    \label{eq:inner-product-condition}
    \frac{\inner{[\vecy], [\vecc]}}{\length{[\vecy]}}
    \geq
    \frac{\sqrt{k - 1}}{1 - \frac{1}{S}
    \left(
      \frac{1}{R^*} + \frac{1}{\sqrt{2R^*}}
    \right)} .
  \end{equation}
  Additionally, the algorithm runs in time polynomial in~$n$, $\log q$
  and~$S$.
\end{theorem}

\begin{proof}
  The proof is identical to that
  of~\cite[Theorem~17]{DBLP:journals/tit/KoetterV03a}, but using the
  fact that both the numerator and denominator of
  $\inner{[\vecy],[\vecc]}/\length{[\vecy]}$ are unchanged when using
  $[\vecc] \in [0,1]^{qn}$ (instead of~$[0,1]^{pn}$) as defined
  in~\cite{DBLP:journals/tit/KoetterV03a} with the appropriate
  (zero-padded) $[\vecy] \in [0,1]^{qn}$, and replacing the inequality
  $\length{[\vecy]}^{2} \geq n/q$ with
  $\length{[\vecy]}^{2} \geq n/2$. The latter inequality holds because
  each block of~$[\vecy]$ has unit~$\ell_{1}$ norm, but has at most
  two nonzero entries.

  The claim on the running time follows from the fact that the
  algorithm runs in time polynomial in~$n$, $\log q$, and the ``cost''
  of the employed multiplicity matrix, which is shown
  in~\cite[Lemma~15]{DBLP:journals/tit/KoetterV03a} to be polynomial
  in~$n$ and~$S$.
\end{proof}

The following immediate corollary gives a more geometric sufficient
condition for a subfield subcode word to be recovered.

\begin{corollary}
  \label{cor:anglecondition}
  Under the same setup as in \cref{thm:inner-product-condition}, a
  subfield subcode word $\vecc \in \cal{C} = \cal{R} \cap \F_{p}^{n}$
  will be in the list output by $\KV([\vecy], S)$ if the angle~$\beta$
  between the vectors $[\vecy], [\vecc] \in [0,1]^{p n}$ satisfies
  \begin{equation}
    \label{eq:anglecondition}
    \cos \beta \geq \frac{\sqrt{R^*}}{1 - \frac{1}{S}
      \parens*{\frac{1}{R^*} + \frac{1}{\sqrt{2R^*}}}} .
  \end{equation}
\end{corollary}

\begin{proof}
  Follows from \cref{thm:inner-product-condition} and the identity
  $\inner{[\vecy], [\vecc]} = \length{[\vecy]} \cdot \length{[\vecc]}
  \cos \beta = \length{[\vecy]} \cdot \sqrt{n} \cos \beta$, where
  $\length{[\vecc]} = \sqrt{n}$ because $\vecc \in \F_{p}^{n}$.
\end{proof}

\begin{theorem}
  \label{thm:euclidean-bch-correctness}
  \cref{alg:euclidean-bch-decoder} is correct. More specifically,
  given input $\vecy \in \R_{p}^n$ and $\varepsilon > 0$, it outputs a
  list of exactly those codewords $\vecc \in \cal{C}$ for which
  $\length{\vecy - \vecc}^{2} \leq (1 - \varepsilon)d/2$, in time
  polynomial in~$n$, $\log q$, and $1/\varepsilon$.
\end{theorem}

\begin{proof}
  The running time follows from the fact that~$S$ is polynomial in~$n$
  and~$1/\varepsilon$ by \cref{eq:list-size-bound}, and by
  \cref{thm:inner-product-condition}.

  We now prove correctness. By the final step, the algorithm outputs only
  codewords $\vecc \in \cal{C}$ for which
  $\length{\vecy - \vecc}^{2} \leq (1-\varepsilon)d/2$. Now
  letting~$\vecc$ be any such a codeword, we show that it appears in
  $\KV([\vecy],S)$, and hence in the output of the overall
  algorithm. From \cref{lem:embedding} and the fact that
  $d = (1-R^{*})n$, we have
  \begin{equation}
    \label{eq:relidist}
    \length{[\vecy] - [\vecc]}^2 \leq 2\length{\vecy - \vecc}^{2}
    \leq (1 - \varepsilon)(1 - R^*)n .
  \end{equation}
  Let $\beta$ be the angle between $[\vecy]$ and $[\vecc]$, which has
  non-negative cosine because both vectors have only non-negative
  entries. Since $\length{[\vecc]}^2 = n$ (because
  $\vecc \in \F_{p}^{n}$), the squared distance from~$[\vecc]$ to the
  line passing through the origin and~$[\vecy]$ is $n \sin^{2} \beta$,
  so
  \begin{equation}
    \label{eq:sinbound}
    n \sin^2 \beta \leq \length{[\vecy] - [\vecc]}^2
  \end{equation}
  and hence $\sin^{2} \beta \leq (1-\varepsilon)(1-R^{*})$. Therefore,
  \begin{equation}
    \label{eq:cosbound}
    \cos \beta = \sqrt{1 - \sin^2 \beta}
    \geq \sqrt{\varepsilon + (1 - \varepsilon)R^*}.
  \end{equation}
  Finally, based on our choice of $S$, a straightforward algebraic
  manipulation yields
  \begin{equation}
    \sqrt{\varepsilon + (1 - \varepsilon)R^*}
    = \frac{\sqrt{R^*}}{1 - \frac{1}{S}\left(\frac{1}{R^*}
        + \frac{1}{\sqrt{2R^*}}\right)} ,
  \end{equation}
  so $\cos \beta$ satisfies \cref{eq:anglecondition}, and invoking
  \cref{cor:anglecondition} completes the proof.
\end{proof}

We conclude this subsection by noting that (list) decoding a linear
code $\cal{C} \subseteq \F_{p}^{n}$ with error measured in the
Euclidean norm is syntactically very similar, but not quite identical,
to (list) decoding the Construction~A lattice
$\Lambda = \overline{\cal{C}} + p\Z^{n} \subseteq \R^{n}$.  Indeed,
when~$p$ exceeds twice the decoding distance, these tasks are
equivalent (and hence \cref{alg:euclidean-bch-decoder} can efficiently
solve the latter) because by the triangle inequality, any given coset
of~$p\Z^{n}$ can have at most one lattice point that is within the
decoding distance of a given received word.  However, when~$p$ is
significantly smaller than~$d$ (as it is in
\cref{sec:near-minkowski}), there may be a huge number of different
lattice vectors in a given coset of~$p\Z^{n}$ that are within the
decoding distance of the received word. So even if it is possible to
decode the code efficiently, it may be (combinatorially) infeasible to
decode its Construction~A lattice.

\subsection{Optimality of Our Reliability Vector}
\label{sec:optimal-reliability}

Here we show that our choice of reliability
vector~$[\vecy] \in [0,1]^{pn}$ as a function of the received
word~$\vecy \in \R_{p}^{n}$ (as used in
\cref{alg:euclidean-bch-decoder}) yields, for the Koetter-Vardy
polynomial-time soft-decision decoder, an essentially optimal tradeoff
between the squared Euclidean decoding distance (up to $n/4$) and the
adjusted rate~$R^{*}$ of the Reed-Solomon code. (The material in this
section is not needed for anything else in the paper.)

As shown in~\cite[Section~IV]{KoetterV01:_optimal_weight} (see
Equation~(12)), the Koetter-Vardy algorithm is guaranteed to decode a
received word $\vecy \in \R_{p}^{n}$ to some desired distance if it is
given a reliability vector $W \in [0,1]^{pn}$ such that
\begin{equation}
  \label{eq:KV-optimal}
  \min_{\vecc} \frac{\inner{[\vecc], W}} {\length{W}} > \sqrt{k-1} =
  \sqrt{nR^{*}} ,
\end{equation}
where the minimum is taken over all $\vecc \in \Z_{p}^{n}$ within the
desired distance of~$\vecy$. So, for provable decoding with this
algorithm, the code's adjusted rate~$R^{*}$ is bounded by the minimum
(over the choice of~$\vecy$) of the maximum (over the choice of~$W$)
left-hand side of \cref{eq:KV-optimal}.

In what follows we will show that for any squared Euclidean decoding
distance~$\delta n$ with $\delta \leq 1/4$, there is a particular
received word $\vecy \in \R_{p}^{n}$ for which our choice of
reliability vector~$W = [\vecy]$ maximizes the left-hand side of
\cref{eq:KV-optimal}, and bounds the adjusted rate by
$R^{*} < 1-2\delta$. Therefore, according to the best available
analysis, the Koetter-Vardy algorithm is limited to decoding to within
squared distance $\delta n < n(1-R^{*})/2 = d/2$. Because our
reliability vector allows for decoding to squared distance
$(1-\varepsilon)d/2$ for any positive constant (or even inverse
polynomial)~$\varepsilon$, it is therefore an essentially optimal
choice.

\paragraph{The analysis.}

Fix some $\delta \in (0,1/4]$ and let $\beta \in (0,1/2]$ be the
unique solution to $\beta(1-\beta) = \delta$. Then let
$\vecy = (\beta, \ldots, \beta) \bmod p$ be the received word
in~$\R_{p}^{n}$ where each entry is $\beta \bmod p$. Also let
$\Delta \in \R^{p}$ be defined by
$\Delta_{\alpha} = \abs{\alpha - \beta}^{2}$ for each
$\alpha \in \Z_{p}$, i.e., the squared Euclidean distance
between~$\alpha$ and~$\beta$ (modulo~$p$).  We proceed by showing
that, for this received word~$\vecy$, taking $W=[\vecy] \in [0,1]^{pn}$
maximizes the left-hand side of \cref{eq:KV-optimal} and makes it
equal $\sqrt{n(1-2\delta)}$, hence $R^{*} < 1-2\delta$.

Because $\vecy$ is the all-$\beta$s vector, to maximize the left-hand
size of \cref{eq:KV-optimal}, without loss of generality we can
take~$W$ to be made up of~$n$ identical blocks $w \in [0,1]^{p}$. For any
word $\vecc \in \Z_{p}^{n}$, the distance $\length{\vecy - \vecc}$ is
entirely determined by the frequencies with which the various
$\alpha \in \Z_{p}$ appear in~$\vecc$. More specifically, define a
vector $T \in [0,1]^{p}$ by
$T_{\alpha} = \frac{1}{n}\sum_{i = 1}^{n} [\vecc_{i} = \alpha]$, the
fraction of entries that are~$\alpha$ in~$\vecc$. Then we have
$\frac{1}{n}\length{\vecy - \vecc}^{2} = \inner{T, \Delta}$. Next,
define the set
\begin{equation}
  \label{eq:error-events}
  B(\delta) = \set{T \in [0,1]^{p} :
    \inner{T, \Delta} \leq \delta,
    \sum_{\alpha} T_{\alpha} = 1,
    T_{\alpha} \geq 0\; \forall \alpha}.
\end{equation}
In~\cite{KoetterV01:_optimal_weight} it is shown (in a more general
form) that taking any $w \in \arg \min_{T \in B(\delta)} \inner{T,T}$
maximizes the left-hand side of \cref{eq:KV-optimal}.

\begin{lemma}
  \label{lem:kkt-conditions}
  For any $ \delta \in (0,1/4]$ and the unique $\beta \in (0,1/2]$
  satisfying $\beta(1-\beta) = \delta$, the reliability vector
  $[\beta] \in [0,1]^{p}$ as defined in \cref{sec:reliability-vectors} is
  the unique element of $\arg \min_{T \in B(\delta)} \inner{T, T}$.
\end{lemma}

\begin{proof}
  Because we are dealing with a convex optimization problem with
  strictly convex objective function $\inner{T,T}$, any local
  minimizer is the unique global minimizer, so it suffices to show
  that~$[\beta]$ is the former. Additionally, because the constraints
  defining $B(\delta)$ are affine, a vector $T \in B(\delta)$ is a
  local minimizer if and only if it satisfies the Karush-Kuhn-Tucker
  conditions
  \begin{align}
    \forall \alpha,\; 2 T_{\alpha} + \mu \Delta_{\alpha}
    - \mu_{\alpha} + \lambda &= 0, \label{eq:kkt-stationarity} \\
    \mu(\inner{T, \Delta} - \delta) &= 0, \label{eq:kkt-comp-slack}
    \\
    \forall \alpha,\; \mu_{\alpha} T_{\alpha} &= 0. \label{eq:kkt-comp-slack-alpha}
  \end{align}
  for some real constants $\mu, \mu_{\alpha} \geq 0$ and
  (unrestricted)~$\lambda$.

  We begin by showing that $[\beta] \in B(\delta)$. Recall from
  \cref{sec:reliability-vectors} that $[\beta]$ has two nonzero
  entries, $[\beta]_{0} = 1 - \beta$ and $[\beta]_{1} = \beta$, so it
  is clear that $[\beta]$ is non-negative with entries that sum to
  one, and it remains to check that
  $\inner{[\beta], \Delta} \leq \delta$. Because $\beta \in (0, 1/2]$,
  we have $\Delta_{0} = \beta^{2}$ and $\Delta_{1} = (1 - \beta)^{2}$,
  and thus
  \[\inner{[\beta], \Delta} = (1 - \beta) \cdot \beta^{2} +
    \beta \cdot (1 - \beta)^{2} = \beta(1 - \beta) = \delta. \]
  Moreover, this shows that \cref{eq:kkt-comp-slack} is satisfied for
  $T=[\beta]$ and any~$\mu \geq 0$.

  Finally, we find $\mu, \mu_{\alpha} \geq 0$ and $\lambda \in \R$ for
  which $T=[\beta]$ satisfies
  \cref{eq:kkt-stationarity,eq:kkt-comp-slack-alpha}.
  Because~$[\beta]_{0}$ and~$[\beta]_{1}$ are the only nonzero entries
  of~$[\beta]$, \cref{eq:kkt-comp-slack-alpha} says that we must take
  $\mu_{0} = \mu_{1} = 0$, but $\mu_{\alpha} \geq 0$ is unrestricted
  for $\alpha \not\in \bit$. We solve the system given by
  \cref{eq:kkt-stationarity} to get
  % \begin{align*}
  %   (1 - \beta) &= [\beta]_{0} = -(\lambda + \mu\beta^{2})/2 \\
  %   \beta &= [\beta]_{1} = -(\lambda + \mu(1 - \beta)^{2})/2.
  % \end{align*}
  $\mu = 2 \geq 0$, $\lambda = -2(\beta^{2} - \beta + 1)$, and for all
  $\alpha \not\in \set{0, 1}$,
  \[ \mu_{\alpha} = \mu\Delta_{\alpha} + \lambda = 2(\Delta_{\alpha} -
    (\beta^{2} - \beta + 1)), \] which is non-negative because
  $\Delta_{\alpha} = \abs{\alpha-\beta}^{2} \geq 1$ and
  $\beta^{2} - \beta \leq 0$ for $\beta \in [0, 1/2]$.
\end{proof}

%%% Local Variables:
%%% mode: latex
%%% TeX-master: "bch-lattice-decoding"
%%% End:

\section{Decoding Construction D Lattices}
\label{sec:construction-d}

In this section we give a list-decoding algorithm for Construction~D
lattices. First we recall the definition of Construction~D.  The
definition we use here is very similar to the one
from~\cite[Section~8.8.1]{conwaysloane99:_splag}, with the only
differences being that we scale so that the lattice is integral (see
\cref{rem:construction-d-rescale}), and we generalize in the obvious
way to codes over~$\F_{p}$ for any prime~$p$.

\begin{definition}[Construction D]
  \label{def:construction-d}
  Let
  $\F_p^n = \cal{C}_0 \supseteq \cal{C}_1 \supseteq \cdots \supseteq
  \cal{C}_\ell$ be a tower of length-$n$ linear codes
  where~$\cal{C}_i$ has dimension~$k_i$ for $i = 0, \ldots, \ell$.
  Choose a basis $\vecb_1, \dots, \vecb_n$ of~$\F_p^n$ such that
  \begin{enumerate*}[label=(\roman*)]
  \item $\vecb_1, \dots, \vecb_{k_i}$ form a basis of~$\cal{C}_i$ for
    $i = 0, \ldots, \ell$, and
  \item some permutation of the (row) vectors
    $\vecb_{1}, \ldots, \vecb_{n}$ forms a upper-triangular
    matrix.\footnote{Without loss of generality, such a basis can be
      obtained by starting with an arbitrary basis of
      $\cal{C}_{\ell}$, extending it to a basis of
      $\cal{C}_{\ell - 1}$, then extending that basis to one of
      $\cal{C}_{\ell - 2}$, and so on
      through~$\cal{C}_{0} = \F_{p}^{n}$. Finally, perform Gaussian
      elimination on the resulting basis of~$\F_{p}^{n}$ so that some
      permutation of the vectors forms an upper-triangular matrix.
      See~\cite[Section~4]{DBLP:journals/toc/Micciancio12} for full
      details.}
  \end{enumerate*}
  Define a set of distinguished $\Z^{n}$-representatives
  for~$\cal{C}_{i}$ as follows: for any $\vecc \in \cal{C}_{i}$, write
  it uniquely as $\vecc = \sum_{j=1}^{k_{i}} a_{j} \vecb_{j}$ for some
  $a_{j} \in \F_{p}$, and define its representative
  $\tilde{\vecc} := \sum_{j=1}^{k_{i}} \overline{a}_{j}
  \overline{\vecb}_{j} \in \Z^{n}$.

  Define $\Lambda_{0} = \Z^{n}$, and for each $i=1,\ldots,\ell$ define
  the integer lattice
  \begin{equation}
    \label{eq:D-Lambda_i}
    \Lambda_{i} := \tilde{\cal{C}}_{i} + p\Lambda_{i - 1}.
  \end{equation}
  The Construction~D lattice for the full tower $\set{\cal{C}_{i}}$
  is~$\Lambda = \Lambda_{\ell}$.
\end{definition}

In \cref{thm:construction-d-props} below (see also
\cref{rem:construction-d-props}) we recall various important
properties of Construction D lattices, which we use there to obtain
our main results.

\begin{remark}
  \label{rem:construction-d-decomp}
  Observe that any vector $\vecv_{i} \in \Lambda_{i}$ can be written
  uniquely as $\vecv_{i} = \tilde{\vecc} + p \vecv_{i-1}$ for some
  $\vecc \in \cal{C}_{i}$ and $\vecv_{i-1} \in \Lambda_{i-1}$. This is
  simply because if
  $\tilde{\vecc} + p \vecv_{i-1} = \tilde{\vecc}' + p \vecv'_{i-1}$
  for some $\vecc, \vecc' \in \cal{C}_{i}$ and
  $\vecv_{i-1}, \vecv'_{i-1} \in \Lambda_{i-1}$, then by reducing
  modulo $p\Z^{n}$, we have $\vecc = \vecc'$ and hence
  $\vecv_{i-1} = \vecv'_{i-1}$ as well.
\end{remark}

\begin{remark}
  \label{rem:construction-d-representative}
  Because the set of representatives~$\tilde{\cal{C}}_{i}$ depends on
  the choice of basis, so do the above lattices~$\Lambda_{i}$. For
  $\vecc = \sum_{j=1}^{k_{j}} a_{j}\vecb_{j} \in \cal{C}_{i}$,
  $\overline{\vecc}$ may differ from $\tilde{\vecc}$ because the
  addition in $\Z^{n}$ does not ``wrap around'' as it does in
  $\F_{p}^{n}$, and their difference may not be a (scaled) lattice
  vector in $p \Lambda_{i-1}$.

  As a concrete example, let $\cal{C}_{0} = \F_{3}^{2}$ and let
  $\cal{C}_{1} = \cal{C}_{2}$ be the code generated by the vector
  $\vecb_{1} = (1, 2) \in \F_{3}^{2}$.\footnote{One can also construct
    a similar example over~$\F_{2}$, but in higher dimension.} Now let
  $\Lambda_{0}, \Lambda_{1}, \Lambda_{2}$ be the lattices obtained via
  Construction D for this tower using basis
  $\vecb_{1}, \vecb_{2} = (0, 1) \in \F_{3}^{2}$, and let
  $\Lambda'_{0}, \Lambda'_{1}, \Lambda'_{2}$ be obtained instead using
  the basis $\vecb'_{1} = (2,1), \vecb_{2} \in \F_{3}^{2}$. Then we
  have $\overline{\vecb}'_{1} \in \Lambda'_{2} \subseteq \Z^{2}$ by
  construction (\cref{eq:D-Lambda_i}). However,
  $\overline{\vecb}'_{1} = (2,1) \not\in \Lambda_{2}$ because
  $\tilde{\vecb}'_{1} = 2 \cdot \overline{\vecb}_{1} = 2 \cdot (1,2) =
  (2,4) \in \Lambda_{2}$ (where we have used the
  representative~$\overline{\vecb}_{1} \in \Z^{2}$ to define
  $\tilde{\vecb}'_{1}$, as required by the construction
  of~$\Lambda_{2}$), but
  $\tilde{\vecb}'_{1} - \overline{\vecb}'_{1} = (0, 3) \not\in
  3\Lambda_{1}$ because $(0, 1) \not\in \cal{C}_{1}$.
\end{remark}

\begin{remark}
  \label{rem:construction-d-rescale}
  Let $\Lambda = \Lambda_{\ell}$ be the lattice obtained via
  \cref{def:construction-d} and let $\Lambda'$ be the lattice obtained
  via~\cite[Section~8.8.1]{conwaysloane99:_splag}, for the same tower
  of (binary) codes. Then $\Lambda = 2^{\ell - 1} \Lambda'$. To see
  this, unwind \cref{def:construction-d} to see that~$\Lambda$
  consists of all vectors of the form
  \[ \vecz + \sum_{i = 1}^{\ell}\sum_{j=1}^{k_{i}} 2^{\ell - i}
    \overline{a}_{j}^{(i)} \overline{\vecb}_{j} \] where
  $\vecz \in 2^\ell\Z^n$ and $a_{j}^{(i)} \in \F_{2}$, which is
  clearly equivalent to $2^{\ell-1}\Lambda'$.
\end{remark}

\begin{remark}
  \label{rem:construction-d-basis}
  Using any basis $\vecb_{1}, \ldots, \vecb_{n}$ of~$\F_{p}^{n}$
  meeting the conditions from \cref{def:construction-d}, we can
  efficiently construct a basis consisting of relatively short vectors
  for the associated Construction D lattice~$\Lambda$. Defining the
  representatives~$\overline{\vecb}_{j}$ to have small entries (e.g.,
  in $\set{0, \ldots, p-1}$), there is a basis of~$\Lambda$ consisting
  of vectors $p^{i_{j}} \overline{\vecb}_{j}$ for various
  $i_{j} \in \set{0, \ldots, \ell}$. (See, e.g., the proof
  of~\cite[Theorem~4.2]{DBLP:journals/toc/Micciancio12}.)  Therefore,
  the vectors in this basis have Euclidean norm at most
  $(p-1) p^{\ell} \sqrt{n}$. For suitable towers of codes~$\cal{C}_{i}$,
  this bound is not much more than the minimum distance of~$\Lambda$;
  see \cref{rem:construction-d-short-basis}.
\end{remark}

% It is useful to think of Construction~D as defnining not only the
% final lattice $\Lambda$ but also intermediary lattices $\Lambda_i$
% obtained by performing Construction~D on the subtower
% $\cal{C}_0 \supseteq \ldots \supseteq \cal{C}_r$. For
% $r = 0, \ldots \ell$, these intermediary lattices all satisfy
% where for convenience we define $\Lambda_{-1} = \Z^n$.

% Added commands to deal with naming of the two different decoders
\newcommand{\eucBchDec}{\cal{D}}
\newcommand{\cstrDDec}{\cal{L}}

The recursive form of \cref{def:construction-d} naturally leads to a
recursive (list) decoder, given in \cref{alg:d} below, which relies on
(list-)decoding subroutines that find
\emph{all}~$\vecc_{i} \in \cal{C}_{i}$ that are sufficiently close to
an (appropriately updated) received word at each stage.
% Let $\F_{p}^{n} = \cal{C}_0 \supseteq \cdots \supseteq \cal{C}_\ell$
% be as in \cref{def:construction-d}.
More precisely, for each $i = 0, \ldots, \ell$ let $\eucBchDec_i$ be a
(list) decoder for the code~$\cal{C}_i$ to
distance~$e_{i} := p^{i} e_0$ (for some $e_{0} > 0$) in some
desired~$\ell_{q}$ norm~$\length{\cdot}$, e.g., the Euclidean norm.
% Additionally, define $e'_{0} = e_{0}$, and for $i = 1, \ldots, \ell$
% let
% \begin{equation}
%   \label{eq:decoding-radius}
%   e'_{i} := \min(p e'_{i - 1}, e_{i}) .
% \end{equation}

\begin{algorithm}[h]
  \caption{List-decoding algorithm~$\cstrDDec(\vecy, i)$ for the
    lattices~$\Lambda_{i}$}
  \label{alg:d}
  \begin{algorithmic}
    \REQUIRE Received word $\vecy \in \R^n$ and integer $i \in \set{0,
    \ldots, \ell}$.

    \ENSURE A list of the lattice vectors
    $\vecv \in \Lambda_i$ for which
    $\length{\vecv - \vecy} \leq e_{i}$.

    \begin{enumerate}[itemsep=0pt]
    \item Let $L = \eucBchDec_i(\vecw) \subseteq \cal{C}_{i}$ where
      $\vecw = \vecy \bmod p\Z^n$.
    \item For each $\vecc \in L$:
      \begin{enumerate}
      % \item Uniquely write
      %   $\vecc = \sum_{j=1}^{k_{i}} \alpha_{j}\vecb_{j}$ for some
      %   $a_{j} \in \F_{p}$, let $\alpha_{j} = \overline{a}_{j}$, and
      %   let
      %   $\vecr = \sum_{j=1} \alpha_{j} \overline{\vecb}_{j} \in
      %   \Z^{n}$.
      \item If $i=0$, let
        $R_{\vecc} = \set{\tilde{\vecc} + p \round{(\vecy - \tilde{\vecc})/p}}
        \subseteq \Z^{n}$.

        (Alternatively, let $R_{\vecc}$ be the list of all elements of
        $\tilde{\vecc} + p\Z^{n}$ that are sufficiently close
        to~$\vecy$.)
      \item Otherwise, let
        $R_{\vecc} = \set{\tilde{\vecc} + p\vecv : \vecv \in \cstrDDec((\vecy -
          \tilde{\vecc})/p, i - 1)} \subseteq \Z^{n}$.
      \end{enumerate}

    \item Output $\bigcup_{\vecc \in L} R_{\vecc}$.
    \end{enumerate}
  \end{algorithmic}
\end{algorithm}

\begin{theorem}
  \label{thm:alg-d-correctness}
  For any $e_{0} < p/2$, \cref{alg:d} is correct: given any
  $\vecy \in \R^{n}$ and $i \in \set{0, \ldots, \ell}$, it outputs a
  list of exactly those $\vecv \in \Lambda_i$ for which
  $\length{\vecy - \vecv} \leq e_{i} = p^{i} e_{0}$.
\end{theorem}

\begin{proof}
  We proceed by induction on~$i$. Starting with the base case $i = 0$,
  by assumption, $D_{0}(\vecw)$ outputs a list~$L$ of exactly those
  $\vecc \in \F_{p}^{n} = \Z_{p}^{n}$ for which
  $\length{\vecw - \vecc} \leq e_{0}$. For each $\vecc \in L$, there
  is a unique element $\vecv \in \tilde{\vecc} + p \Z^{n}$ for which
  $\length{\vecy-\vecv} \leq e_{0}$, because $e_{0} < p/2$ and by the
  triangle inequality. Indeed,
  $\vecv = \tilde{\vecc} + p \round{(\vecy-\tilde{\vecc})/p}$ is that
  unique element, because it minimizes the magnitude of each
  coordinate of $\vecy-\vecv$. Therefore, $\cstrDDec(\vecy,0)$ outputs a
  list of exactly those $\vecv \in \Lambda_{0}$ for which
  $\length{\vecy-\vecv} \leq e_{0}$, as claimed.

  Next, assume by induction that the algorithm correctly list
  decodes~$\Lambda_{i-1}$ to distance~$e_{i-1}$ for some $i \geq
  1$. First, it is clear that $\cstrDDec(\vecy, i)$ outputs \emph{only}
  vectors within distance~$e_{i}$ of~$\vecy$: since
  $\cstrDDec((\vecy-\tilde{\vecc})/p, i-1)$ outputs only vectors~$\vecv$
  for which $\length{(\vecy-\tilde{\vecc})/p - \vecv} \leq e_{i-1}$ by
  assumption, we have
  $\length{\vecy - (\tilde{\vecc} + p \vecv)} \leq p e_{i-1} = e_{i}$
  for all the vectors $\tilde{\vecc} + p \vecv \in R_{\vecc}$, as
  needed.

  Finally, we show that $\cstrDDec(\vecy, i)$ outputs a list containing
  \emph{all} $\vecv \in \Lambda_{i}$ for which
  $\length{\vecy - \vecv} \leq e_{i}$. Let~$\vecv$ be such a
  vector. By \cref{rem:construction-d-decomp} we can uniquely write
  $\vecv = \tilde{\vecc} + p \vecv_{i-1}$ for some
  $\vecc \in \cal{C}_{i}$ and $\vecv_{i-1} \in \Lambda_{i-1}$. Then,
  because $\length{\vecy - \vecv} \leq e_{i}$, we also have
  $\length{\vecw - \vecc} \leq e_{i}$, so $\vecc$ appears in the list
  output by $\eucBchDec_{i}(\vecw)$. Finally, we have that
  \begin{equation}
    \length*{(\vecy - \tilde{\vecc})/p - \vecv_{i-1}} = \length{(\vecy -
      (\tilde{\vecc} + p \vecv_{i-1}))/p} \leq e_{i}/p = e_{i-1},
  \end{equation}
  so $\vecv_{i-1}$ appears in the list output by
  $\cstrDDec((\vecy - \tilde{\vecc})/p, i-1)$. Thus, $\vecv$ appears in
  the list output by $\cstrDDec(\vecy, i)$, as needed.
\end{proof}

\begin{theorem}
  \label{thm:alg-d-runtime}
  Let~$S_{i}, T_{i} \geq 1$ be upper bounds on the output list size
  and running time, respectively, of the decoder $\eucBchDec_i$, and let
  $S=\max_{i} S_{i}$ and $T = \max_{i} T_{i}$. Then the running time
  $R(i)$ of $\cstrDDec(\cdot, i)$ satisfies
  \begin{equation}
    \label{eq:alg-d-runtime-bound}
    R(i) \leq
    (i+1) (T+K) \prod_{j=1}^{i} S_{j} \leq (i+1) (T+K) \cdot S^{i}
  \end{equation}
  for some $K = \poly(S,n,\log p)$.
\end{theorem}

\begin{proof}
  The execution of $\cstrDDec(\vecy, i)$ consists of one call to
  $\eucBchDec_i$, at most $S_i$ calls to $\cstrDDec(\cdot, i - 1)$, and
  some $K=\poly(S,n,\log p)$ additional work. So $R(i)$ satisfies the
  recurrence
  \begin{equation}
    R(i) \leq S_i \cdot R(i-1) + (T_i + K)
  \end{equation}
  with the initial value $R(0) = T_0 + K$. The bound from
  \cref{eq:alg-d-runtime-bound} follows immediately by unwinding this
  recurrence.
\end{proof}

%%% Local Variables:
%%% mode: latex
%%% TeX-master: "bch-lattice-decoding"
%%% End:

\section{Decoding Near Minkowski's Bound}
\label{sec:near-minkowski}

In this section we recall a certain family of $n$-dimensional
Construction D lattices whose (normalized) minimum distances are
within an $O(\sqrt{\log n})$ factor of optimal, and instantiate
\cref{alg:d} to efficiently list decode those lattices to within a
$1/\sqrt{2} - \varepsilon$ factor of the minimum distance, for any
constant $\varepsilon > 0$.

Throughout this section, for simplicity we restrict our focus to codes
over characteristic-two fields. We can get similar results for larger
characteristic $p > 2$ by generalizing \cref{thm:construction-d-props}
in the natural way. However, because BCH codes for larger
characteristic have a weaker codimension bound (see
\cref{lem:bch-dimension}), the corresponding family of lattices admit
a weaker bound on their normalized minimum distances.

\begin{theorem}[{{\cite[Chapter 8, Theorem 13, rescaled]{conwaysloane99:_splag}}}]
  \label{thm:construction-d-props}
  Let
  $\F_2^n = \cal{C}_0 \supseteq \cal{C}_1 \supseteq \cdots \supseteq
  \cal{C}_\ell$ be a tower of length-$n$ binary linear codes where
  $\cal{C}_i$ has dimension~$k_{i}$ and minimum Hamming distance
  $d_i \geq 4^i$ for $i = 0, \ldots, \ell$. The Construction~D
  lattice~$\Lambda = \Lambda_{\ell}$ for the tower $\set{\cal{C}_i}$
  has Euclidean minimum
  distance\footnote{In~\cite{conwaysloane99:_splag} only a lower bound
    on the minimum distance is claimed, but it is easy to see that
    this is an equality, since~$\Lambda_{\ell}$ has $2^{\ell} \Z^{n}$
    as a sublattice.}  $\lambda_1(\Lambda) = 2^{\ell}$ and determinant
  \begin{equation}
    \label{eq:construction-d-determinant}
    \det(\Lambda) = 2^{n\ell - \sum_{i = 1}^{\ell} k_i} =
    2^{\sum_{i=1}^{\ell} (n-k_{i})}.
  \end{equation}
\end{theorem}

\begin{remark}
  \label{rem:construction-d-props}
  The determinant of the Construction D lattice~$\Lambda$ follows
  directly from the fact that the (row) vectors $\overline{\vecb}_{i}$
  can be permuted to form an upper-triangular matrix, and scaling them
  by appropriate powers of two yields a basis of~$\Lambda$.

  As intuition (but not a proof) for the minimum distance
  of~$\Lambda$, this arises from the fact that the minimum Hamming
  distance of the codes increases by a factor of four at each
  level. This means that the minimum Euclidean norm of the vectors in
  the representative sets~$\tilde{\cal{C}}_{i}$ increases by a factor
  of two at each level, because for integer vectors, the Euclidean
  norm is at least the square root of the Hamming weight. These
  varying Euclidean minimum distances are then equalized by scaling
  the layers~$\Lambda_{i}$ by corresponding powers of two.
\end{remark}

\begin{remark}
  \label{rem:construction-d-short-basis}
  Using \cref{rem:construction-d-basis}, we can obtain a basis
  of~$\Lambda$ whose vectors have Euclidean norm at most
  $2^{\ell} \sqrt{n}$. For a tower of codes~$\cal{C}_{i}$ meeting the
  conditions from \cref{thm:construction-d-props}, the minimum
  distance of~$\Lambda$ is~$2^{\ell}$. Therefore, the lengths of the
  basis vectors are within a~$\sqrt{n}$ factor of optimal. We can use
  this basis to generate many more relatively short lattice vectors,
  by taking small integer linear combinations. This can be used for
  encoding messages as short lattice vectors, as mentioned in the
  introduction.
\end{remark}

\begin{construction}[BCH lattice family]
  \label{con:dense-lattice}
  Let~$q$ be a power of two, let $n = q - 1$, and let
  $\ell \leq \log_{4} n$ be a positive integer. For each
  $i = 0, \ldots, \ell$ let $\cal{C}_{i} = \cal{C}_{\F_{q}}[n, 4^{i}]$
  be the BCH code of length~$n$ with designed distance
  $d_{i} = 4^{i} \leq n$. Define the lattice~$\Lambda_{q,\ell}$ to be
  the Construction~D lattice for the tower
  $\F_{2}^{n} = \cal{C}_{0} \supseteq \cal{C}_{1} \supseteq \cdots
  \supseteq \cal{C}_{\ell}$.\footnote{Recall from
    \cref{sec:preliminaries} that
    $\cal{C}_{i-1} \supseteq \cal{C}_{i}$ because
    $d_{i-1} \leq d_{i}$, so the codes form a tower, as required.}
\end{construction}

This construction and the following lemma are essentially the same as
the ones that appear in
\cite{barnes83:_lattice_packings,DBLP:journals/toc/Micciancio12}. The
construction could also be performed with towers of \emph{extended}
BCH codes of length $n=q$, in which a parity-check bit is appended to
every BCH codeword. That alternative construction achieves the same
asymptotic bound on the normalized minimum distance, as shown
in~\cite{DBLP:journals/toc/Micciancio12}.

\begin{lemma}
  \label{lem:dense-lattice-hermite}
  For any $q = 2^{\kappa}$ with $n=q-1$ and
  $\ell = \log_{4} \Theta(n/\log n)$, the $n$-dimensional lattice
  $\Lambda = \Lambda_{q,\ell}$ satisfies
  $\lambda_{1}(\Lambda)/\det(\Lambda)^{1/n} = \Omega(\sqrt{n/\log
    n})$.
\end{lemma}

\begin{proof}
  Let $h := 2^{\ell}$.  By \cref{thm:construction-d-props}, $\Lambda$
  has minimum distance $\lambda_{1}(\Lambda) \geq h$ and determinant
  $\det(\Lambda) = 2^{\sum_{i=1}^{\ell} (n - k_{i})}$. Using the
  codimension bound from \cref{lem:bch-dimension}, for all $i \geq 1$
  we have $n - k_{i} \leq \kappa (4^{i}/2)$. Substituting this bound
  and expanding the resulting geometric series, we get
  \begin{equation*}
    \sum_{i=1}^{\ell} (n - k_{i}) \leq \sum_{i=1}^{\ell} \kappa (4^{i}/2)
    = \kappa (4^{\ell+1}/6) = 2\kappa h^{2}/3.
  \end{equation*}
  Thus we have
  $\det(\Lambda) \leq 2^{2\kappa h^{2}/3} = q^{2h^{2}/3}$. Finally,
  using $h = 2^{\ell} = \Theta(\sqrt{n/\log n})$, this yields
  \begin{equation*}
    \frac{\lambda_{1}(\Lambda)}{\det(\Lambda)^{1/n}} \geq \frac{h}{q^{2h^{2}/(3n)}}
    = \Omega(h) = \Omega(\sqrt{n/\log n}).
  \end{equation*}
\end{proof}

\begin{theorem}
  \label{thm:dense-lattice-decoder}
  Let~$q$ be a power of two with $n=q-1$, and let
  $\ell \leq \log_{4} (n-1)$. \cref{alg:d}, using
  \cref{alg:euclidean-bch-decoder} as the decoder~$\cal{D}_{i}$ for
  code~$\cal{C}_{i}$, is a list-decoding algorithm for the
  lattice~$\Lambda = \Lambda_{q,\ell}$ defined in
  \cref{con:dense-lattice}. Specifically, on input a received word
  $\vecy \in \R^{n}$ and $\varepsilon > 0$, the algorithm outputs a
  list of exactly those $\vecv \in \Lambda$ such that
  $\length{\vecy - \vecv} \leq \lambda_{1}(\Lambda) \cdot \sqrt{(1 -
    \varepsilon)/2}$, in time $\poly(n,1/\epsilon)^{\ell}$. Moreover,
  for any $\ell \leq \log_{4}((1-\Omega(1))n)$ and constant
  $\varepsilon > 0$, the algorithm runs in time $\poly(n)$.
\end{theorem}

\begin{proof}
  Because $d_{i} = 4^{i}$, by \cref{thm:euclidean-bch-correctness},
  $\cal{D}_{i}$ decodes~$\cal{C}_{i}$ to Euclidean distance
  $2^{i}\sqrt{(1 - \varepsilon)/2}$. So, the decoders~$\cal{D}_{i}$
  satisfy the requirements of \cref{alg:d}, with
  $e_{0} = \sqrt{(1-\varepsilon)/2} < 1$. Therefore, by
  \cref{thm:alg-d-correctness} and \cref{thm:construction-d-props},
  \cref{alg:d} list decodes~$\Lambda$ to Euclidean distance
  $2^{\ell} e_{0} = \lambda_{1}(\Lambda) \cdot
  \sqrt{(1-\varepsilon)/2}$.

  By \cref{thm:euclidean-bch-correctness} and
  \cref{thm:alg-d-runtime}, the running time is
  $\poly(n, 1/\varepsilon) \cdot S^{\ell}$, where
  $S=\poly(n,1/\varepsilon)$ is as in \cref{eq:list-size-bound}, for
  $1/R^{*} \leq n/(n-d_{\ell}) \leq n$. In particular, when
  $d_{\ell} = 4^{\ell} \leq (1-\Omega(1))n$, we have $1/R^{*} = O(1)$,
  and when $\varepsilon > 0$ is a positive constant, we have
  $S = O(1)$ by \cref{eq:list-size-bound}. Therefore, because
  $\ell=O(\log n)$, the running time
  $\poly(n,1/\varepsilon) \cdot S^{\ell} = \poly(n)$, as claimed.
\end{proof}

Finally, taking any $\ell = \log_{4} \Theta(n/\log n)$ and combining
\cref{lem:dense-lattice-hermite} with \cref{thm:dense-lattice-decoder}
yields our main result, \cref{thm:main}.

%%% Local Variables:
%%% mode: latex
%%% TeX-master: "bch-lattice-decoding"
%%% End:

\bibliography{common/lattices.bib,common/codes.bib,common/crypto.bib}
\bibliographystyle{common/alphaabbrvprelim}

\end{document}